\documentclass[journal]{IEEEtran}
\usepackage{amssymb}
\usepackage{amsfonts}
\usepackage{savesym}
\usepackage{amsfonts,subfigure,multicol,color,verbatim,
graphicx,cite,epsfig,amssymb,amsmath,cases,bm,algorithm,
algorithmic,xcolor,multirow}
\usepackage{amsmath}
\usepackage{array}
\usepackage{enumerate}

\newtheorem{problem}{\textbf{Problem}}

\newtheorem{theorem}{\textbf{Theorem}}
\newtheorem{definition}{\textbf{Definition}}
\newtheorem{lemma}{\textbf{Lemma}}
\newtheorem{corollary}{\textbf{Corollary}}
\newtheorem{alg}{\textbf{Algorithm}}
\newtheorem{step}{\textbf{Step}}
\newtheorem{remark}{\textbf{Remark}}
\newtheorem{notation}{\textbf{Notation}}

\begin{document}
\markboth{IEEE Systems Journal} {Li et al: MDP in C-RAN\ldots}

\title{Resource Allocation Optimization for Delay-Sensitive Traffic in Fronthaul Constrained Cloud Radio Access Networks}

\author{Jian~Li,~Mugen~Peng, ~\IEEEmembership{Senior Member,~IEEE},
~Aolin~Cheng,~Yuling~Yu,~Chonggang~Wang, ~\IEEEmembership{Senior
Member,~IEEE}
\thanks{Manuscript received June 18, 2014; revised September 12, 2014; accepted October 15, 2014. The editor coordinating the review of this paper and approving it for publication was Prof. Vincenzo Piuri.}
\thanks{Jian~Li (e-mail: {\tt lijian.wspn@gmail.com}), Mugen~Peng (e-mail: {\tt pmg@bupt.edu.cn}), Aolin~Cheng (e-mail: {\tt ahchengaolin@gmail.com}), ~Yuling~Yu (e-mail: {\tt aliceyu1215@gmail.com}) are with the Key Laboratory of Universal
Wireless Communications for Ministry of Education, Beijing
University of Posts and Telecommunications, China. Chonggang~Wang
(e-mail: {\tt cgwang@ieee.org}) is with the InterDigital
Communications, King of Prussia, PA, USA.}
\thanks{This work was supported in part by the National Natural Science Foundation of China (Grant No. 61222103, No.61361166005), the National High Technology Research and Development Program of China (Grant No. 2014AA01A701), the State Major Science and Technology Special Projects (Grant No. 2013ZX03001001), and the Beijing Natural Science Foundation (Grant No. 4131003).}
}

\maketitle
\begin{abstract}
The cloud radio access network (C-RAN) provides high spectral and
energy efficiency performances, low expenditures and intelligent
centralized system structures to operators, which has attracted
intense interests in both academia and industry. In this paper, a
hybrid coordinated multi-point transmission (H-CoMP) scheme is
designed for the downlink transmission in C-RANs, which fulfills the
flexible tradeoff between cooperation gain and fronthaul
consumption. The queue-aware power and rate allocation with
constraints of average fronthaul consumption for the delay-sensitive
traffic are formulated as an infinite horizon constrained partially
observed Markov decision process (POMDP), which takes both the
urgent queue state information (QSI) and the imperfect channel state
information at transmitters (CSIT) into account. To deal with the
\emph{curse of dimensionality} involved with the equivalent Bellman
equation, the linear approximation of post-decision value functions
is utilized. A stochastic gradient algorithm is presented to
allocate the queue-aware power and transmission rate with H-CoMP,
which is robust against unpredicted traffic arrivals and
uncertainties caused by the imperfect CSIT. Furthermore, to
substantially reduce the computing complexity, an online learning
algorithm is proposed to estimate the per-queue post-decision value
functions and update the Lagrange multipliers. The simulation
results demonstrate performance gains of the proposed stochastic
gradient algorithms, and confirm the asymptotical convergence of the
proposed online learning algorithm.
\end{abstract}

\begin{IEEEkeywords}
Queue-aware resource allocation, hybrid coordinated multi-point
transmission, fronthaul limitation, cloud radio access networks.
\end{IEEEkeywords}

\section{Introduction}
\IEEEPARstart{I}{t} is estimated that the demand for high-speed
mobile data traffic, such as high-quality wireless video streaming,
social networking and machine-to-machine communication, will get
1000 times increase by 2020{\cite{bib:tdscdma}}, which requires a
revolutionary approach involving new wireless network architectures
as well as advanced signal processing and networking technologies.
As key components of heterogeneous networks (HetNets), low power
nodes (LPNs) are deployed within the coverage of macro base stations
(MBSs) and share the same frequency band to increase the capacity of
cellular networks in dense areas with high traffic demands.
Unfortunately, the aggressive reuse of limited radio spectrum will
result in severe inter-cell interference and unacceptable
degradation of system performances. Therefore, it is critical to
control interference through advanced signal processing techniques
to fully unleash the potential gains of HetNets. As an integral part
of the LTE-Advanced (LTE-A) standards, the coordinated multi-point
transmission (CoMP) technique targets the suppression of the
inter-cell interference and quality of service (QoS) improvement for
the cell-edge UEs. However, CoMP is faced with some disadvantages
and challenges in real HetNets. The performance gain of CoMP highly
depends on the perfect knowledge of channel state information (CSI)
and the tight synchronization, both of which pose strict
restrictions on the backhaul of LPNs. To manipulate the high density
of LPNs with lowest capital expenditure (CAPEX) and operational
expenditure (OPEX) effectively, the cloud radio access network
(C-RAN) was proposed in{\cite{cran}} to enhance spectral efficiency
and energy efficiency performances and has recently attracted
intense interest in both academia and industry.

\begin{figure}
\centering \vspace*{0pt}
\includegraphics[scale=0.85]{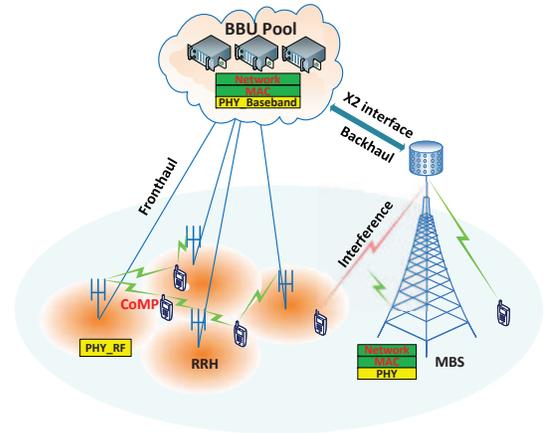}
\setlength{\belowcaptionskip}{-700pt} \vspace*{-5pt}\caption{C-RAN
architecture} \label{fig1}
\end{figure}

As depicted in Fig. \ref{fig1}, the remote radio heads (RRHs) are
only configured with the front radio frequency (RF) and simple
symbol processing functionalities, while the other baseband physical
processing and procedures of the upper layers are executed jointly
in the baseband unit (BBU) pool for UEs associating with RRHs. The
LPNs are simplified as RRHs through connecting to a ``signal
processing cloud" with high-speed fronthaul links. To coordinate the
cross-tier interference between RRHs and MBSs effectively, the BBU
pool is interfaced to MBSs. Such a distributed deployment and
centralized processing architecture facilitates the implementation
of CoMP{\cite{crancluster}} amongst RRHs of C-RANs as well as
provides ubiquitous networks coverage with MBSs. Since all the RRHs
in C-RANs are connected to the BBU pool, the CoMP can be realized
through virtual beamforming and the beamformers can be calculated in
BBU pool. Specifically, the CoMP in downlink C-RANs can be
characterized into two classes{\cite{comp}}: joint processing (JP)
and coordinated beamforming (CB). For JP, the traffic payload is
shared and transmitted jointly by all RRHs within the CoMP
cluster{\cite{JP}}, which means multiple delivery of the same
traffic payload from the centralized BBU pool to each cooperative
RRH through capacity-limited fronthaul links. As for the CB, the
traffic payload is only transmitted by the serving RRH, but the
corresponding beamformer is jointly calculated at the centralized
BBU pool to coordinate the interference to all other UEs within the
CoMP cluster{\cite{CS}}. Obviously, JP achieves higher average
spectrum efficiency than CB does at the expense of more fronthaul
consumption, while CB requires more antennas equipped with each RRH
to achieve the full intra-cluster interference coordination.
However, the practical non-ideal fronthaul with limited capacity
restricts the overall performances of CoMP in C-RANs.

\subsection{Related Works}

There exists lots of literatures aiming to alleviate the fronthaul
requirement of JP without the loss of interference exploitation. The
authors of{\cite{dynamicclustering}} proposed a dynamic clustered
multi-cell cooperation scheme to substantially reduce the backhaul
consumption by imposing restriction on the cluster size. A heuristic
algorithm was proposed in{\cite{directional}} to dynamically select
the directional cooperation links under a finite-capacity backhaul
subject to the evaluation of benefits and costs, which cannot
completely eliminate the undesired interference as in the full
cooperation case. Both reweighed $l$-1 norm minimization method and
heuristic iterative link removal algorithm were proposed
in{\cite{tonyquek}} to reduce the user data transfer via the
capacity-limited backhaul effectively by dealing with the formulated
cooperative clustering and beamforming problems, which, however, are
suboptimal and still suffer from a significant performance loss. A
backhaul cost metric considering the number of active directional
cooperation links was adopted in{\cite{asymmetricJP}}, where the
design problem is minimizing this backhaul cost metric and jointly
optimizing the beamforming vectors among the cooperative BSs subject
to signal-to-interference-and-noise-ratio (SINR) constraints at UEs.
To make a flexible tradeoff between the cooperation gain and the
backhaul consumption, a rate splitting approach under the limited
backhaul rate constraints was proposed in{\cite{streamsplitting}},
where some fraction of the backhaul capacity originally consumed by
JP could be privately used to get more performance gains. Borrowing
the idea in{\cite{streamsplitting}}, the authors of{\cite{r1}}
proposed a soft switching strategy between the JP-CoMP and CB-CoMP
modes under capacity-limited backhaul. Considering the high
complexity and the large signaling overhead, a distributed hard
switching strategy was also proposed in{\cite{r1}}. To achieve a
tradeoff between diversity and multiplexing gains of multiple
antennas and high spectral efficiency, the authors of{\cite{r2}}
studied both the dynamic partial JP-COMP and its corresponding
resource allocation in a clustered CoMP cellular networks.
Generally, for the C-RANs, the potential high spectral efficiency
gain of CoMP largely depends on the quality of obtained channel
state information at transmitters (CSIT) as well as the fronthaul
consumption. In{\cite{r3}}, channel prediction usefulness was
analyzed and compared with channel estimation in downlink CoMP
systems with backhaul latency in time-varying channels, considering
both the centralized and decentralized JP-CoMP as well as the
CB-CoMP.

However, the aforementioned works only focus on physical layer
performance of spectral efficiency or energy efficiency and ignore
the bursty traffic arrival as well as the delay requirement of
delay-sensitive traffic. Therefore, the resulting control policy is
adaptive to the channel state information (CSI) only and cannot
guarantee good delay performance for delay-sensitive applications.
In general, since the CSI could provide information regarding the
channel opportunity while the queue state information (QSI) could
indicate the urgency of the traffic flows, the queue-aware resource
allocation should be adaptive to both the CSI and QSI. Furthermore,
as the CSIT cannot be perfect in real systems, and systematic packet
errors occur when the allocated data rate exceeds the instantaneous
mutual information. Therefore, the issue of robustness against the
uncertainty incurred by imperfect CSIT should also be considered in
the resource allocation optimization.

There already have some research efforts on the queue-aware dynamic
resource allocation in stochastic wireless networks. In
paper{\cite{queuecluster}}, the authors proposed a mixed timescale
delay-optimal dynamic clustering and power allocation design with
downlink JP in traditional multi-cell networks. The queue-aware
discontinuous transmission (DTX) and user scheduling design with
downlink CB in energy-harvesting multi-cell networks was proposed
in{\cite{dtx}}. A queue-weighted dynamic optimization algorithm
using Lyapunov optimization approach was proposed
in{\cite{lyapunov}} for the joint allocation of subframes, resource
blocks, and power in the relay-based HetNets. However, all these
works focus on queue-aware resource allocations in homogeneous
networks or HetNets without the consideration of the imperfect CSIT.
Therefore, the solutions cannot work in the C-RANs with the
practical challenges of imperfect CSIT and non-ideal
capacity-limited fronthaul links.

\subsection{Main Contributions}

To the best of our knowledge, there are lack of effective signal
processing techniques and dynamic radio resource management
solutions for delay-sensitive traffic in C-RANs to optimize the SE,
EE and delay performances, which still remains challenging and
requires more investigations. Based on the aforementioned advantages
and challenges of C-RANs, the efficient CoMP scheme with tradeoff
between cooperation gain and fronthaul consumption will be
elaborated in this paper. Furthermore, under the average power and
fronthaul consumption constraints, the dynamic radio resource
management with feature of queue-awareness to maintain good delay
performance for delay-sensitive traffic in stochastic C-RANs will
also get studied in this paper. The major contributions of this
paper are as follows.

\begin{itemize}
\item To allow a flexible tradeoff between cooperation gain and average fronthaul
consumption, the H-CoMP scheme is proposed for the delay-sensitive
traffic in C-RANs by splitting the traffic payload into shared
streams and private streams. By reconstructing the shared streams
and private streams and optimizing the precoders and decorrelators,
the shared streams and private streams can be simultaneously
transmitted to obtain the maximum achievable degree of freedom (DoF)
under limited fronthaul consumption.

\item Motivated by{\cite{survey}}, to minimize the
transmission delay of the delay-sensitive traffic under the average
power and fronthaul consumption constraints in C-RANs, the
queue-aware rate and power allocation problem is formulated as an
infinite horizon average cost constrained partially observed Markov
process decision (POMDP). The queue-aware resource allocation policy
is adaptive to both QSI and CSIT in the downlink C-RANs and can be
obtained by solving a per-stage optimization for the observed system
state at each scheduling frame.

\item Since the optimal solution requires centralized implementation and perfect knowledge
of CSIT statistics and has exponential complexity w.r.t. the number
of UEs, the linear approximation of post-decision value functions
involving POMDP is presented, based on which a stochastic gradient
algorithm is proposed to allocate power and transmission rate
dynamically with low computing complexity and high robustness
against the variations and uncertainties caused by unpredictable
random traffic arrivals and imperfect CSIT. Furthermore, the online
learning algorithm is proposed to estimate the post-decision value
functions effectively.

\item The delay performances of the proposed H-CoMP and queue-aware
resource allocation solution are numerically evaluated. Simulation
results show that a significant delay performance gain can be
achieved in the fronthaul constrained C-RANs with H-CoMP, and the
queue-aware resource allocation solution is validated and effective
due to the adaptiveness to both QSI and imperfect CSIT. Further, the
stochastic gradient algorithms can improve the delay performances
drastically, and the online learning algorithm is asymptotically
converged.
\end{itemize}

The rest of this paper is organized as follows. Section II describes
the system model and section III gives the design of H-CoMP scheme
for the downlink C-RANs. The queue-aware resource allocation problem
is formulated as POMDP in section IV and a low complexity approach
is proposed in section V. The performance evaluation is conducted in
section VI and section VII summarizes this paper.

\begin{notation}
$(.)^T$ and $(.)^H$ stand for the transpose and conjugate transpose,
respectively. $(.)^\dagger $ stands for the pseudo-inverse. Besides,
$diag(\textbf{p})$ denotes a diagonal matrix formed by the vector
$\textbf{p}$.
\end{notation}

\section{System Models}

To optimize performances of downlink C-RANs, the transmission model,
traffic queue dynamic model in the medium access control (MAC)
layer, and the imperfect CSIT assumption in the physical layer are
considered in this section.

\subsection{Transmission Model}

The transmission of $M$ delay-sensitive traffic payloads in downlink
C-RANs with $M$ RRHs is considered. Denote $\mathcal {M} =
\{1,2,...,M\}$ as the UE set and $\mathcal {N} = \{1,2,...,M\}$ as
the RRH set within the CoMP cluster. An example of C-RAN with $M =
2$ is illustrated in Fig. 2. The inter-tier interferences amongst
the RRHs and MBSs are controlled by setting the maximum allowable
power consumption of each RRH indicated by the MBSs through the X2
interfaces, while the intra-tier interferences in C-RANs can be
eliminated by implementing the CoMP. Each RRH and UE are equipped
with $N_t$ and $N_r$ antennas respectively, where $M{N_r} \geq N_t >
(M - 1){N_r}$. Within the coverage of each RRH, a served UE exists
and it can also be cooperatively served by the other RRHs according
to the following proposed H-CoMP scheme. In this paper, the
scheduling is carried out in every frame indexed by $t$ and the
frame duration is $\tau$ second.

\begin{figure}
\centering \vspace*{0pt}
\includegraphics[scale=0.50]{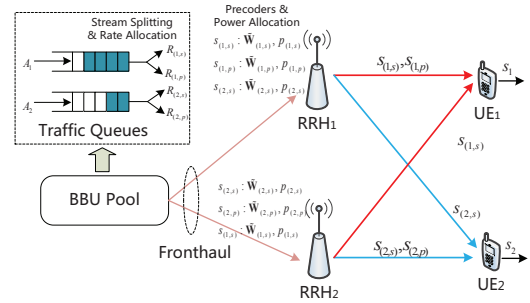}
\setlength{\belowcaptionskip}{-700pt}
\vspace*{-5pt}\caption{Workflow of resource allocation for $M$ = 2.}
\end{figure}

\subsection{Traffic Queue Dynamic Model}

Let ${\bf{Q}}(t) = \{ {Q_1}(t), \ldots ,{Q_M}(t)\}$ denote the
global QSI (number of bits) for $M$ queues maintained at BBU pool at
the beginning of scheduling frame $t$. There will be random packet
arrival ${A_i}(t)$ after ${G_i}(t)$ bits are successfully received
by UE $i$ at the end of frame $t$. The random arrival process
$A_i(t)$ is supposed to be independent identically distributed
(i.i.d) over scheduling frame according to a general distribution
with mean $\mathbb{E}\{ {A_i}(t)\} = {\lambda _i}$ and independent
w.r.t $i$. Furthermore, the statistics of ${A_i}(t)$ is supposed to
be unknown to the BBU. The queue dynamic of UE $i$ is then given by
\begin{equation}
{Q_i}(t + 1) = min\{ {[{Q_i}(t) - {G_i}(t)]^ + } + {A_i}(t),{N_Q}\},
\end{equation}
where $[x]^+ = max\{x, 0\}$ and $N_Q$ is the maximum buffer size.

\subsection{Imperfect CSIT Assumption}

Let ${{\bf{H}}_{ji}}(t) \in {\mathbb{C}^{{N_r} \times {N_t}}}$
denote the complex channel fading coefficient between RRH $i$ and UE
$j$ at frame $t$ and let $\textbf{H}(t) = \{{{\bf{H}}_{ji}}(t): j\in
\mathcal {M}, i \in \mathcal {N}\}$ denote the global CSI.
Especially, every element of ${{\bf{H}}_{ji}}(t)$ is supposed to
remain constant within a scheduling frame but be i.i.d over
scheduling frame. The perfect knowledge of CSI is assumed to be only
obtained by the UE while the imperfect CSIT ${\bf{\hat H}} =
\{{{\bf{\hat H}}_{ji}}(t) \in {\mathbb{C}^{{N_r} \times {N_t}}}:
j\in \mathcal {M}, i \in \mathcal {N}\}$ is obtained by the BBU
pool. The rank of both ${\bf{\hat H}}_{ji}$ and ${\bf{H}}_{ji}$ is
assumed to be $min\{N_r,N_t\}$. Furthermore, the imperfect CSIT
error kernel model is given by{\cite{imperfectCSIT}}

\begin{equation}
\Pr [{{\bf{\hat H}}_{ji}}|{{\bf{H}}_{ji}}] = \frac{1}{{\pi {\sigma
_{ji}}}}\exp ( - \frac{{|{{{\bf{\hat H}}}_{ji}} -
{{\bf{H}}_{ji}}{|^2}}}{{{\sigma _{ji}}}}),
\end{equation}
which is caused by duplexing delay in time division duplex (TDD)
systems or quantization errors and feedback latency in frequency
division duplex (FDD) systems. The above ${\sigma _{ji}} \in [0,1]$
indicates the CSIT quality. When ${\sigma _{ji}} = 0$, we have
${{\bf{\hat H}}_{ji}}={{\bf{H}}_{ji}}$, which corresponds to the
perfect CSIT case. When ${\sigma _{ji}} = 1$, we have ${{\bf{\hat
H}}_{ji}}{{\bf{H}}_{ji}^{\dag }}$ = {\textbf{0}}, which corresponds
to the no CSIT case.

\section{Hybrid CoMP scheme}

With the limited fronthaul capacity, the maximum achievable DoF can
be obtained by separating the traffic payload for UE $i$ into
${L_{(i,s)}}$ shared streams ${{\bf{s}}_{(i,s)}}$ and ${L_{(i,p)}}$
private streams ${{\bf{s}}_{(i,p)}}$ and simultaneously transmitting
them with optimal precoders and decorrelators, that is the hybrid
CoMP (H-CoMP) scheme. More specifically, the H-CoMP allows shared
streams to be shared across the RRHs with the CoMP cluster by
multiple delivery through capacity-limited fronthaul links.
Meanwhile, the H-CoMP makes private streams remain private to
certain RRH and the precoders are jointly calculated at the BBU pool
to eliminate the intra-cluster interference. Therefore, the
cooperative transmission of shared streams requires significantly
more fronthaul consumption than the coordinated transmission of
private streams does. In the following subsections, the traffic
streams splitting model, precoder calculation and decorrelator
calculation with the perfect CSIT will be thoroughly elaborated.

\subsection{Traffic Streams Splitting Model}

To make a flexible tradeoff between the cooperation gain and average
fronthaul consumption, the number of shared streams and private
streams should be determined with the traffic streams splitting
model. With the perfect CSIT, the zero-forcing (ZF) precoder and
decorrelator designs are adopted for both shared streams and private
streams. In this situation, at most ${L_{M,{N_t},{N_r}}} = {{N_t} -
(M - 1){N_r} }$ private streams can be zero-forced at RRH $i$ to
eliminate interference to UE ${\rm{j}} \ne i$, i.e.
\begin{equation}
L_{(i,p)} \leq {L_{M,{N_t},{N_r}}}.
\end{equation}
Furthermore, to fully recover the ${L_{(i,p)}}$ private streams and
${L_{(i,s)}}$ shared streams at UE $i$, the constraint
\begin{equation}
L_{(i,p)}+L_{(i,s)} \le N_r
\end{equation}
should be satisfied. With the traffic streams splitting, the
proposed H-CoMP allows a flexible tradeoff between the cooperation
gain and the fronthaul consumption. The achievable DoFs of different
schemes are compared in table I.

\begin{table}[tbp]
\centering \caption{Comparison of Achievable DoFs of Different
Schemes}
\begin{tabular}{|m{1.5cm}<{\centering}|m{5cm}<{\centering}|} 
\hline \textbf{Scheme} & \textbf{Achievable DoF }  \\ \hline CB-CoMP
& $(M
-1){L_{M,{N_t},{N_r}}} + {N_r}$ \\ \hline JP-CoMP & $M{N_r}$\\
\hline H-CoMP & $\{ (M - 1){L_{M,{N_t},{N_r}}} + {N_r},...,M{N_r}\}$\\
\hline
\end{tabular}
\end{table}

Specifically, when $L_{(i,p)}+L_{(i,s)} = N_r$, the maximum DoF of
$MN_rM$ can be achieved by the H-CoMP scheme, and when ${L_{(i,p)}}
= {L_{M,{N_t},{N_r}}}$, the fronthaul consumption is minimized.

Although the shared streams and private streams are superimposed in
the downlink transmission of C-RANs, it is possible to eliminate the
interference at RRHs and recover both of them at UEs by constructing
the private streams and shared streams and designing optimal
precoders and decorrelators.

Let ${{\bf{s}}_{(i,s)}} \in {\mathbb{C}^{L_{(i,s)} \times 1}}$ and
${{\bf{s}}_{(i,p)}} \in {\mathbb{C}^{L_{(i,p)} \times 1}}$ denote
the shared streams and private streams respectively, where
$L_{(i,s)}$ and $L_{(i,p)}$ are the number of shared streams and
private streams. To facilitate the implementation of H-CoMP, the
shared streams and private streams are reconstructed by inserting
zero vectors as follows respectively
\begin{equation}
{{\bf{\tilde s}}_{(i,s)}} = {\{ {\bf{s}}_{(i,s)}^T,{{\textbf{0}}_{1
\times {L_{(i,p)}}}}\} ^T},
\end{equation}
\begin{equation}
{{\bf{\tilde s}}_{(i,p)}} = {\{ {\textbf{0}_{1 \times
{L_{(i,s)}}}},{\bf{s}}_{(i,p)}^T\} ^T}.
\end{equation}

Let ${{\bf{\tilde W}}_{(i,s)}} \in {\mathbb{C}^{M{N_t} \times
{(L_{(i,s)} + L_{(i,p)})}}}$ and ${\tilde{\bf{ W}}_{(i,p)}} \in
{\mathbb{C}^{{N_t} \times {(L_{(i,p)} + L_{(i,s)})}}}$ denote the
precoders for the reconstructed shared streams and reconstructed
private streams of UE $i$ respectively. Define ${{\bf{\Lambda
}}_{(i,s)}} = diag(\sqrt {P_{(i,s)}^1} , \ldots ,\sqrt
{P_{(i,s)}^{{L_{(i,s)}}}}, {\textbf{0}_{1 \times {L_{(i,p)}}}})$ and
${{\bf{\Lambda }}_{(i,p)}} = diag({\textbf{0}_{1 \times
{L_{(i,s)}}}}, \sqrt {P_{(i,p)}^1} , \ldots ,\sqrt
{P_{(i,p)}^{{L_{(i,p)}}}} )$, where $P_{(i,s)}$ and $P_{(i,p)}$
denote the transmission power of each shared stream and private
stream for UE $i$ respectively. Then the received signal vector
${{\bf{r}}_i} \in {\mathbb{C}^{{N_r} \times 1}}$ at UE $i$ is given
by
\begin{eqnarray}
&&{{\bf{r}}_i} = \underbrace {{{\bf{H}}_i}{{\tilde{\bf{
W}}}_{(i,s)}}{{\bf{\Lambda }}_{(i,s)}}{{\tilde{\bf{s}}}_{(i,s)}} +
{{\bf{H}}_{ii}}{{\tilde{\bf{ W}}}_{(i,p)}}{{\bf{\Lambda
}}_{(i,p)}}{{\tilde{\bf{
s}}}_{(i,p)}}}_{{\rm{the~desired~signals~for~UE~i }}} \\\nonumber &&
+ \underbrace {\sum\limits_{j \ne i} {{{\bf{H}}_i}{{\tilde{\bf{
W}}}_{(j,s)}}{{\bf{\Lambda }}_{(j,s)}}{{\tilde{\bf{s}}}_{(j,s)}}}
}_{{\rm{the~interference~from~shared~ streams}}}\\\nonumber && +
\underbrace {\sum\limits_{j \ne i} {{{\bf{H}}_{ij}}{{\tilde{\bf{
W}}}_{(j,p)}}{{\bf{\Lambda }}_{(j,p)}}{{\tilde{\bf{s}}}_{(j,p)}}}
}_{{\rm{the~interference~from~ private~streams}}} + {{\bf{n}}_i},
\end{eqnarray}
where ${\textbf{H}_i} = [\begin{array}{*{20}{c}}
{{\textbf{H}_{i1}}}& \cdots &{{\textbf{H}_{iM}}}
\end{array}] \in {\mathbb{C}^{{N_r} \times M{N_t}}}$ is the aggregate complex
channel fading coefficient vector from the $M$ cooperative RRHs to
UE $i$, and ${{\bf{n}}_i}$ is the zero-mean unit variance complex
Gaussian channel noise at UE $i$.

\subsection{Precoder and Decorrelator Calculation for Shared Streams }

The optimal cooperative precoder at the $M$ RRHs and the
decorrelator at UE $i$ should be designed to maximize the mutual
information of shared streams and to eliminate the interference
imposed on all the other UEs (UE $j \ne i$) as follows
\begin{eqnarray}
&&\{ {\bf{\tilde W}}_{_{(i,s)}}^ * ,{\bf{\tilde U}}_{_{(i,s)}}^ * \}
= \arg \mathop {\max }\limits_{{{{\bf{\tilde W}}}_{(i,s)}},
{{{\bf{\tilde U}}}_{(i,s)}}} {\log _2}\det [\textbf{I} + \nonumber \\
&&~~~~~~~~~~~~~~~~~~~~~~~ {{{\bf{\tilde
U}}}_{(i,s)}}{{\bf{{\bf{H}}}}_i} {{{\bf{\tilde
W}}}_{(i,s)}}{\bf{\tilde W}}_{_{(i,s)}}^H{\bf{H}}_{_i}^H{\bf{\tilde
U}}_{_{(i,s)}}^H] \nonumber\\
&&s.t.~~~~~~ {{\bf{H}}_j}{{\bf{\tilde W}}_{(i,s)}} = \textbf{0}.
\end{eqnarray}
Therefore, the precoder have the form of

\begin{equation}
{{\bf{\tilde W}}_{(i,s)}^*} = {{\bf{F}}_{(i,s)}}{{{\bf{\tilde
V}}}_{(i,s)}},
\end{equation}
where ${{\bf{F}}_{(i,s)}} \in {\mathbb{C}^{M{N_t} \times (M{N_t} -
(M - 1){N_r})}}$ is given by the orthonormal basis of
${\rm{nullspace([}}{\bf{H}}_1^T, \cdots ,{\bf{H}}_{i - 1}^T,
{\bf{H}}_{i + 1}^T{\rm{,}} \cdots {\rm{,}}$
${\bf{H}}_M^T{{\rm{]}}^T})$. Let ${{\bf{H}}_i}{{\bf{F}}_{(i,s)}} =
{{\bf{U}}_{(i,s)}}{{\bf{\Sigma }}_{(i,s)}}{\bf{V}}_{(i,s)}^H$ be the
singular value decomposition (SVD) of equivalent channel matrix
${{\bf{H}}_i}{{\bf{F}}_{(i,s)}}$, where the singular values in
${{\bf{\Sigma }}_{(i,s)}}$ are sorted in a decreasing order along
the diagonal, ${{\bf{\tilde V}}_{(i,s)}} \in {\mathbb{C}^{(M{N_t} -
(M - 1){N_r}) \times {(L_{(i,s)}+L_{(i,p)})}}}$ is then given by the
first ${L_{(i,s)}}$ columns of ${{\bf{V}}_{(i,s)}}$ as ${{\bf{\tilde
V}}_{(i,s)}} = \{ {\bf{v}}_{(i,s)}^1, \cdots
,{\bf{v}}_{(i,s)}^{{L_{(i,s)}}},{\textbf{0}}_{(M{N_t} - (M -
1){N_r}) \times L_{(i,p)}}\}$. Furthermore, the decorrelator
${\bf{\tilde U}}_{_{(i,s)}}^ *$ is given by the first ${L_{(i,s)}}$
columns of ${{\bf{U}}_{(i,s)}}$ as follows

\begin{equation}
{\bf{\tilde U}}_{_{(i,s)}}^ * = \{ {\bf{u}}_{(i,s)}^1, \cdots
,{\bf{u}}_{(i,s)}^{{L_{(i,s)}}},{\textbf{0}}_{{N_r} \times
L_{(i,p)}}\}^H.
\end{equation}
Then the recovered $L_{(i,s)}$ shared streams are given by the first
$L_{(i,s)}$ rows of ${\bf{\tilde U}}_{_{(i,s)}}^ *{{\bf{r}}_i}$.

\subsection{Precoder and Decorrelator Calculation for Private Streams }

The optimal coordinated precoder at RRH $i$ for the private streams
of UE $i$ and the decorrelator at UE $i$ should be designed to
maximize the mutual information of private streams and to eliminate
the interference imposed on all the other UEs (UE $j \ne i$) as
follows
\begin{eqnarray}
&&\{ {\bf{\tilde W}}_{_{(i,p)}}^ * ,{\bf{\tilde U}}_{_{(i,p)}}^ * \}
= \arg \mathop {\max }\limits_{{{{\bf{\tilde W}}}_{(i,p)}},
{{{\bf{\tilde U}}}_{(i,p)}}} {\log _2}\det [\textbf{I} + \nonumber\\
&&~~~~~~~~~~~~~~~~~{{{\bf{\tilde U}}}_{(i,p)}}{{\bf{{\bf{H}}}}_{ii}}
{{{\bf{\tilde W}}}_{(i,p)}}{\bf{\tilde W}}_{_{(i,p)}}^H{\bf{H}}_{ii}^H{\bf{\tilde U}}_{_{(i,p)}}^H] \nonumber\\
&&s.t.~~~~~~~ {{\bf{H}}_{ji}}{{\bf{\tilde W}}_{(i,p)}} = \textbf{0}.
\end{eqnarray}
Therefore the precoder has the similar form of
\begin{equation}
{{\bf{\tilde W}}_{(i,p)}^*} = {{\bf{\tilde F}}_{(i,p)}}{{{\bf{\tilde
V}}}_{(i,p)}},
\end{equation}
where ${\bf{\tilde F}}_{(i,p)} = [{\textbf{0}}_{N_t \times (N_r -
L_{M,N_t,N_r})}, {\bf{F}}_{(i,p)}]$ and ${\bf{F}}_{(i,p)} \in
{\mathbb{C}^{{N_t} \times L_{M,N_t,N_r}}}$ is given by the
orthonormal basis of ${\rm{nullspace([}}{\bf{H}}_{1i}^T, \cdots
,{\bf{H}}_{i-1 i}^T, {\bf{H}}_{i + 1 i}^T{\rm{,}} \cdots
{\rm{,}}{\bf{H}}_{Mi}^T{{\rm{]}}^T})$. Let
${{\bf{H}}_{ii}}{{\bf{\tilde F}}_{(i,p)}} =
{{\bf{U}}_{(i,p)}}{{\bf{\Sigma }}_{(i,p)}}{\bf{V}}_{(i,p)}^H$ be the
SVD of equivalent channel matrix ${{\bf{H}}_{ii}}{{\bf{\tilde
F}}_{(i,p)}}$, where the singular values in ${{\bf{\Sigma
}}_{(i,p)}}$ are sorted in an increasing order along the diagonal,
${{\bf{\tilde V}}_{(i,p)}} \in {\mathbb{C}^{{N_r} \times
{(L_{(i,s)}+L_{(i,p)})}}}$ is then given by the last ${L_{(i,p)}}$
columns of ${{\bf{V}}_{(i,p)}}$ as ${{\bf{\tilde V}}_{(i,p)}} =
\{{\textbf{0}}_{{N_r} \times L_{(i,s)}},{\bf{v}}_{(i,p)}^{N_r -
L_{(i,p)} + 1}, \cdots ,{\bf{v}}_{(i,p)}^{{N_r}}\}$. Furthermore,
the decorrelator ${\bf{\tilde U}}_{_{(i,p)}}^ *$ is given by the
last ${L_{(i,p)}}$ columns of ${{\bf{U}}_{(i,p)}}$ as follows
\begin{equation}
{\bf{\tilde U}}_{_{(i,p)}}^ * = \{{\textbf{0}}_{{N_r} \times
L_{(i,s)}}, {\bf{u}}_{(i,p)}^{N_r - L_{(i,p)} + 1}, \cdots
,{\bf{u}}_{(i,p)}^{N_r}\}^H.
\end{equation}
Then the recovered $L_{(i,p)}$ private streams are given by the last
$L_{(i,p)}$ rows of ${\bf{\tilde U}}_{_{(i,p)}}^
*{{\bf{r}}_i}$.

\begin{remark}{(\emph{The Interference Nulling Between Shared Streams and Private
Streams})} Although the interference nulling constraints are not
explicitly imposed, the interference between shared streams and
private streams of UE $i$ can be still eliminated due to the fact
that ${\bf{\tilde U}}_{_{(i,s)}}^
* {{\tilde{\bf{ W}}}_{(i,p)}}{{\bf{\Lambda
}}_{(i,p)}}{{\tilde{\bf{s}}}_{(i,p)}} = \textbf{0}$ and ${\bf{\tilde
U}}_{_{(i,p)}}^
* {{\tilde{\bf{ W}}}_{(i,s)}}{{\bf{\Lambda
}}_{(i,s)}}{{\tilde{\bf{s}}}_{(i,s)}} = \textbf{0}$.
\end{remark}

\subsection{The Power Consumption and Transmission Rate}

To support the cooperative transmission of $a$-th shared stream from
$M$ RRHs to UE $j$, the power contributed by RRH $i$ is given by
$P_{(j,s)}^a\rho _{(j,i)}^a$, where $P_{(j,s)}^a$ denote the total
power to transmit the $a$-th shared stream to UE $i$ and $\rho
_{(j,i)}^a = \sum\nolimits_{x = 1}^{{N_t}} {|{{[{{\bf{\tilde
W}}^{*}_{(j,s)}}]}_{((i - 1){N_t} + x,a)}}{|^2}}$ denote the
contribution by RRH $i$. To support the coordinated transmission of
$a$-th private stream from RRH $i$ to UE $i$, the power of
$P_{(i,s)}^a$ is needed. Therefore, with the proposed H-CoMP scheme,
the total transmit power consumption at RRH $i$ is given by
\begin{equation}
P_i = \sum\nolimits_{a = 1}^{{L_{(i,p)}}} {P_{(i,p)}^a} +
\sum\nolimits_{j = 1}^M {\sum\nolimits_{a = 1}^{{L_{(j,s)}}}
{P_{(j,s)}^a\rho _{(j,i)}^a} }\label{power}.
\end{equation}

In practice, both the precoders and decorrelators are calculated at
the BBU pool with the imperfect CSIT, which will cause uncertain
residual interference to the recovered streams. By treating the
uncertain interference as noise, the mutual information of $a$-th
shared stream at ${\bf{s}}_{(i,s)}^a$ UE $i$ is given by
\begin{equation}
C_{(i,s)}^a = {\log _2}(1 + {{\varphi}_{(i,s)}^a}P_{(i,s)}^a/(1 +
I_{(i,s)}^a)),
\end{equation}
where $\varphi _{(i,s)}^a = |{\bf{\tilde
U}}_{(i,s)}^a{{\bf{H}}_i}{\bf{\tilde W}}_{(i,s)}^a{|^2} $,
${\bf{\tilde U}}_{(i,s)}^a$ and ${\bf{\tilde W}}_{(i,s)}^a$ is the
$a$-th row of ${\bf{\tilde U}}_{(i,s)}^*$ and $a$-th column of
${\bf{\tilde W}}_{(i,s)}^*$ respectively, $I_{(i,s)}^a$ is the
residual interference incurred by the imperfect CSIT. The mutual
information of $a$-th private stream ${\bf{s}}_{(i,p)}^a$ of UE $i$
is given in a similar way. Due to the uncertainties of mutual
information, the data rate successfully transmitted to UE $i$ is
given by
\begin{equation}
{G_i} = ({R_{(i,s)}}\textbf{1}({R_{(i,s)}} \le {C_{(i,s)}}) +
{R_{(i,p)}}\textbf{1}({R_{(i,p)}} \le {C_{(i,p)}}))\tau,
\end{equation}
where ${C_{(i,s)}} = \sum\nolimits_{a = 1}^{{L_{(i,s)}}}
{C_{(i,s)}^a} $ and ${C_{(i,p)}} = \sum\nolimits_{a =
1}^{{L_{(i,p)}}} {C_{(i,p)}^a}$ are the mutual information for
shared streams and private streams respectively. ${R_{(i,s)}} $ and
${R_{(i,p)}}$ are the allocated data rate for shared streams and
private streams of UE $i$ respectively.

\section{Formulation of Queue-aware Control Problem}

To meet the urgency of the delay-sensitive traffic payloads and
reduce the occurrence of packet transmission failure in the downlink
C-RANs, the queue-aware resource allocation problem based on the
observed system states (QSI and
 CSIT) will be formulated in this section.

\subsection{Feasible Stationary Control Policy}

Considering the inter-tier interference imposed by RRHs and the
energy efficient transmission of delay-sensitive traffic, the
feasible resource allocation policy should satisfy the following
average power consumption constraints
\begin{equation}
{P_i}{\rm{(}}\Omega {\rm{) = }}\mathop {lim}\limits_{T \to \infty }
{sup} \frac{1}{T}\sum\limits_{t = 1}^T {{{\mathbb{E}}^\Omega
}[{P_i}(t)]} \le P_i^{\max },
\end{equation}
where ${\mathbb{E}^\Omega }$ indicates that the expectation is taken
w.r.t the measure induced by policy $\Omega$, ${P_i}(t)$ is the
total power consumption of RRH $i$ to support the H-CoMP
transmission, and $P_i^{\max }$ is the maximum average power
consumption indicated by MBSs. Furthermore, by varying the maximum
average power consumption of each RRH, cross-tier interference could
be well controlled to maintain desirable average QoS requirement for
macro UEs.

It is worth noting that compared with the fronthaul consumption for
traffic payload sharing, that for signaling delivery is negligible.
Due to the capacity-limited fronthaul links of C-RANs, the feasible
resource allocation policy also should satisfy the following average
fronthaul consumption constraints
\begin{equation}
R_i^{f}(\Omega )\!=\!\mathop {lim} \limits_{T \to \infty }\!{sup}
\frac{1}{T}\sum\limits_{t = 1}^T \!{{{\mathbb{E}}^\Omega
}[{R_{(i,p)}}(t)\! +\! \sum\limits_j\! {{R_{(j,s)}}(t)} ]} \! \le\!
R_i^{\max },
\end{equation}
where ${R_{(i,p)}}(t) + \sum\nolimits_{j \in \mathcal {M}}
{{R_{(j,s)}}(t)}$ is the total data rate to be delivered to RRH $i$
through the fronthaul link connecting RRH $i$ to BBU pool, and
$R_i^{\max }$ is the maximum average fronthaul consumption.

With the aforementioned resource constraints, the feasible
stationary resource allocation policy for C-RANs is defined as
follows.

\begin{definition}{(\emph{Stationary Resource Allocation Policy})}
A feasible stationary resource allocation policy $\Omega ({\bf{\hat
S}}) = \{ {\Omega _R}({\bf{\hat S}}),{\Omega _P}({\bf{\hat S}})\}$
is a mapping from the global observed system states ${\bf{\hat S}} =
\{ {\bf{Q}},{\bf{\hat H}}\}$ instead of the global system states
${\bf{S}} = \{ {\bf{Q}},{\bf{H}},{\bf{\hat H}}\}$ to the resource
allocation actions, where ${\Omega _P}({\bf{\hat S}}) = \{
{{{P}}^a_{(i,p)}},{{{P}}^b_{(i,s)}}: 1 \leq a \leq L_{(i,p)}, 1 \leq
b \leq L_{(i,s)}, i \in M\}$ and ${\Omega _R}({\bf{\hat S}}) = \{
{{{R}}_{(i,p)}},{{{R}}_{(i,s)}}:i \in M\}$ are the power allocation
policy and rate allocation policy subject to average power
consumption constraints and average fronthaul consumption
constraints.
\end{definition}

Given the feasible stationary resource allocation policy $\Omega
({\bf{\hat S}})$, the induced random process ${\bf{S}} = \{
{\bf{Q}},{\bf{H}},{\bf{\hat H}}\}$ is a controlled Markov chain with
the transition probability as follows
\begin{eqnarray}
\Pr\!\{ {\bf{S}}(t\! +\! 1)|{\bf{S}}(t),\Omega ({\bf{\hat
S}}(t))\}\! = \Pr\{ {\bf{\hat H}}(t \!+ \!1),{\bf{H}}(t \!+\! 1)\}
\nonumber\\ \Pr\{ {\bf{Q}}(t\! + \!1)|{\bf{S}}(t),\Omega ({\bf{\hat
S}}(t))\} .\label{tansprob}
\end{eqnarray}

Apparently, the queue dynamics of the $M$ UEs served by C-RAN are
coupled with each other via $\Omega ({\bf{\hat S}}(t))$.

\subsection{Problem Formulation}

With the positive weighting factors $\{{\beta}_i\}$ ,which indicates
the relative importance of delay requirement among the $M$ users,
the queue-aware resource allocation problem with average power
consumption constraints and average fronthaul consumption
constraints can be formulated as the following problem.

\begin{problem}{(\emph{Queue-aware Resource Allocation Problem})}
 \vspace*{-2pt}
\begin{eqnarray}
&&{\min _\Omega }D({\bf{\beta }},\Omega ) = \mathop {lim} \limits_{T \to \infty } {sup}\frac{1}{T}\sum\limits_{t = 1}^T {{{\mathbb{E}}^\Omega }[\sum\limits_{i \in M} {{\beta _i}\frac{{{Q_i}}}{{{\lambda _i}}}} ]} \\
&& s.t. \rm{~~the~constraints~(17)~and~(18)~for~each~RRH} \nonumber,
\end{eqnarray}
where the $\frac{{{Q_i}}}{{{\lambda _i}}}$ in objective function is
the average traffic delay cost for UE $i$ by Little's Law.
\end{problem}

With the average power consumption constraints and average fronthaul
consumption constraints in Problem 1, the occurrence of extreme
instantaneous power and fronthaul consumption tends to be
impossible. Furthermore, the feasible stationary resource allocation
policy is defined on the observed system states ${\bf{\hat S}} = \{
{\bf{Q}},{\bf{\hat H}}\}$. Therefore, problem 1 is a constrained
partially observed MDP (POMDP){\cite{bertsekas}}, which will be
solved by the following general approach.

\subsection{General Approach with MDP}
Using the Lagrange duality theory, the Lagrange dual function of
problem 1 is defined as

\begin{equation}
J(\gamma )\! = \!\mathop {\min }\limits_\Omega L(\beta ,\gamma
,\Omega ({\bf{\hat S}}) ) \!=\! \mathop {lim}\limits_{T \to \infty
}{sup} \frac{1}{T}\!\sum\limits_{t = 1}^T {{{\mathbb{E}}^\Omega
}[g({\bf{\beta }},{\bf{\gamma }},\Omega ({\bf{\hat
S}}))]}\label{lag},
\end{equation}
where $g({\bf{\beta }},{\bf{\gamma }},\Omega ({\bf{\hat S}})) =
\sum\limits_i {({\beta _i}\frac{{{Q_i}}}{{{\lambda _i}}}) + }
{\gamma _{(i,P)}}({{\rm{P}}_i} - P_i^{\max }) + {\gamma
_{(i,R)}}(R_i^{f} - R_i^{\max })$ is the per-stage system cost and
${\gamma _{(i,P)}}$ and ${\gamma _{(i,R)}}$ are the non-negative
Lagrange multipliers (LMs) w.r.t the power consumption constraints
and fronthaul consumption constraints. Then the dual problem of
problem 1 is given by
\begin{equation}
\mathop {\max }\nolimits_\gamma J(\gamma ).
\end{equation}

Although (\ref{lag}) is an unconstrained POMDP, the solution is
generally nontrivial. To substantially reduce the global observed
system states space, the partitioned actions are defined as follows
with the i.i.d. property of the CSIT.

\begin{definition}{(\emph{Partitioned Actions})}
Given the stationary resource allocation policy $\Omega$, $\Omega
({\bf{Q}}) = \{ \Omega ({\bf{\hat S}}):\forall {\bf{\hat H}}\}$ is
defined as the collections of power and rate allocation actions for
all possible CSIT ${\bf{\hat H}}$ on a given QSI $\textbf{Q}$,
therefore $\Omega$ is equal to the union of all partitioned actions.
i.e. $\Omega = {\cup}_{\textbf{Q}}{{\Omega}(\textbf{Q})} $.
\end{definition}

As the distribution of the traffic arrival process is unknown to the
BBU, the post-decision state potential function instead of potential
function will be introduced in the following theorem to derive the
queue-aware resource allocation policy of eq. (\ref{lag}).

\begin{theorem}{(\emph{Equivalent Bellman Equation})}\\
(a)Given the LMs, the unconstrained POMDP problem can be solved by
the equivalent Bellman equation as follows
\begin{eqnarray}
U({\bf{\tilde Q}}) + \theta = \sum\nolimits_{\bf{A}} {\Pr
({\bf{A}})} \mathop {min}\limits_{\Omega(\textbf{Q})}g(\bf{\beta },
{\bf{\gamma }}, \textbf{Q},{\Omega(\textbf{Q})}) \nonumber\\ +
\sum\nolimits_{{\bf{\tilde Q}}'} {\Pr \{ {\bf{\tilde
Q}}'|\textbf{Q},{\Omega(\textbf{Q})}\} U({\bf{\tilde Q}}')}
\label{bellman},
\end{eqnarray}
where $g(\bf{\beta }, {\bf{\gamma }},{\bf{Q}},\Omega ({\bf{Q}})) =
\mathbb{E}[g({\bf{\beta }},{\bf{\gamma }},\Omega ({\bf{\hat
S}}))|{\bf{Q}}] $ is the conditional per-stage cost and $\Pr \{
{\bf{\tilde Q}}'|{\bf{Q}},{\Omega(\textbf{Q})}\} = \mathbb{E}[\Pr
[{\bf{Q}}'|{\bf{H}},{\bf{Q}},\Omega ({\bf{\hat S}})]|{\bf{Q}}]$ is
the conditional average transition kernel, $U({\bf{\tilde Q}})$ is
the post-decision value function. ${\bf{\tilde Q}}$ is the
post-decision state and ${\bf{\tilde Q}}' = {({\bf{Q}} - {\bf{G}})^
+ }$ is the next post-decision state, where ${\bf{Q}} = \min \{
{\bf{\tilde Q}} + {\bf{A}},{N_Q}\}$ and ${\bf{G}} = \{G_i : i \in
\mathcal {M}\}$.

(b)If there exists unique $(\theta ,\{ U({\bf{\tilde Q}})\} )$ that
satisfies (\ref{bellman}), then $\theta = \mathop
{min}\limits_{\Omega(\textbf{Q})} {\mathbb{E}}[g(\bf{\beta },
{\bf{\gamma }},{\bf{Q}},\Omega ({\bf{Q}}))$ is the optimal average
per-stage cost for the unconstrained POMDP and the optimal resource
allocation policy $\Omega $ is obtained by minimizing R.H.S of
(\ref{bellman}).
\end{theorem}

\begin{proof}
Please refer to Appendix A
\end{proof}

\begin{remark}{(\emph{The Zero Duality Gap})}
Although the objective function of problem 1 is not convex w.r.t the
stationary resource control policy, the duality gap between the dual
problem and primal problem is zero when the condition Theorem 1 (b)
is established, which implies that the primal optimal resource
control policy can be obtained by solving the equivalent Bellman
equation of the dual optimal problem.
\end{remark}

\begin{remark}{(\emph{The Computational Complexity})}
Solving the equivalent Bellman equation involves ${N_Q}^M + 1$
unknowns $(\theta ,\{ U({\bf{\tilde Q}})\} )$ and ${N_Q}^M$
nonlinear fixed point equations, which means exponential state
space, enormous computational complexity and full knowledge of
system states transition probability in (\ref{tansprob}). Therefore,
a low complexity solution based on linear approximation and online
learning of post-decision value functions will be further studied.
\end{remark}

\section{Low Complexity Approach}

In this section, to substantially reduce the enormous computing
complexity in centralized BBU pool, the linear approximation of
post-decision value functions is utilized, upon which a stochastic
gradient algorithm is proposed to obtain the QAH-CoMP policy and an
online learning algorithm is proposed to estimate the post-decision
value functions.

\subsection{Linear Approximation of Post-decision Value Functions}

The linear approximation of post-decision value functions is defined
by the sum of the per-queue value functions as
follows{\cite{linearapprox}}

\begin{equation}
U({\bf{\tilde Q}}) \approx \sum\nolimits_{i \in \mathcal {M}}
{{U_i}({{\tilde Q}_i})},
\end{equation}
where ${U_i}({\tilde Q_i})$ is the per-queue post-decision value
functions which satisfies the following per-queue fixed point
Bellman equation
\begin{eqnarray}
{U_i}({{{\rm{\tilde Q}}}_i})\! + \!{\theta _i}\!&& =
\!\sum\nolimits_{{A_i}} {\Pr ({A_i})} \mathop
{min}\limits_{{\Omega}_{i}(Q_i)}[{g_i}({\beta _i},{{\bf{\gamma }}_i}
,{Q_i},{{\Omega}_{i}})\nonumber\\&& + \sum\nolimits_{{{\tilde
Q}_i}'} {\Pr \{ \tilde Q{'_i}|{Q_i},{{\Omega}_{i}}\} U(\tilde
Q{'_i})}] \label{perqueuebellman},
\end{eqnarray}
where ${g_i}({\beta _i},{{\bf{\gamma }}_i} ,{Q_i},{{\Omega}_{i}}) =
\mathbb{E}[ {\beta _i}\frac{{{Q_i}}}{{{\lambda _i}}} + {\gamma
_{(i,P)}}(\sum\limits_{a = 1}^{{L_{(i,s)}}} {P_{(i,s)}^a\rho
_{(i,i)}^a} + \sum\limits_{a = 1}^{{L_{(i,p)}}} {P_{(i,p)}^a} -
P_i^{\max }) + \sum\limits_{j \in \mathcal {M}, j \ne i} {{\gamma
_{(j,P)}}}$ $ \sum\limits_{a = 1}^{{L_{(i,s)}}} {P_{(i,s)}^a\beta
_{(i,j)}^a} + {\gamma _{(i,R)}}({R_{(i,p)}}(t) + \sum\limits_{j \in
\mathcal {M}} {{R_{(j,s)}}(t)} - R_i^{\max })|Q_i]$ is the per-queue
per-stage cost function. ${Q_i} = \min \{ {\tilde Q_i} +
{A_i},{N_Q}\}$ is the pre-decision state and $\tilde Q{'_i} =
{({Q_i} - {G_i})^\dag }$ is the next post-decision state. The
optimality of linear approximation is established in the following
lemma.

\begin{lemma}{(\emph{The Optimality of Linear Approximation})}
The linear approximation is optimal only when the CSIT is perfect,
which means the interference is completely eliminated with H-CoMP
scheme, therefore, the queue dynamics of $M$ UEs are decoupled.
\end{lemma}

\begin{proof}
Please refer to Appendix B.
\end{proof}

Generally, the error variance ${\sigma}_{ji}$ of the imperfect CSIT
can not be large, therefore the linear approximation is
asymptotically accurate with sufficiently small error variance of
CSIT.

\begin{remark}{(\emph{The Computing Complexity})}
With the linear approximation, the calculation of the post-decision
value functions in BBU pool is alleviated from exponential
complexity $ \mathcal {O}((N_Q + 1)^M)$ to polynomial complexity
$\mathcal {O}((N_Q + 1)M)$.
\end{remark}

\subsection{Low Complexity QAH-CoMP Policy}

With the combination of the linear approximation and equivalent
Bellman equation (\ref{bellman}), the QAH-CoMP policy can be
obtained by solving the following per-stage optimization for every
observed system state, which is summarized as the following
corollary.

\begin{corollary}[Per-Stage Optimization]
With the observation of current system states, the per-stage
optimization is given by
\begin{equation}
{\Omega ^*}({\bf{\hat S}}) \!=\! \{ \Omega _P^*({\bf{\hat
S}}),\Omega _R^*({\bf{\hat S}})\} \! =\! {\rm{arg}}\!\!\!\mathop
{min}\limits_{\Omega _P({\bf{\hat S}}),\Omega _R({\bf{\hat
S}})}\!\!\! B({\bf{\hat S}},\!{\Omega _P({\bf{\hat S}}),\Omega
_R({\bf{\hat S}})}),
\end{equation}
where $B({\bf{\hat S}},{\Omega _P({\bf{\hat S}}),\Omega _R({\bf{\hat
S}})}) = \sum\limits_{i \in \mathcal {M}} {\{{\gamma _{(i,P)}}P_i
{\gamma _{(i,R)}}({R_{(i,p)}}(t) }$ ${+ \sum\limits_j
{{R_{(j,s)}}(t)} ) +
{\mathbb{E}}[{{\bf{1}}_{(i,s)}}{{\bf{1}}_{(i,p)}}]({U_i}({Q_k} -
{\tau R_{(i,s)}}}{ - \tau {R_{(i,p)}})}$ ${ - {U_i}({Q_k} - \tau
{R_{(i,p)}}) - {U_i}({Q_k} - \tau {R_{(i,s)}}) + {U_i}({Q_k}))} +
{\mathbb{E}}[{{\bf{1}}_{(i,s)}}]({U_i}({Q_k} - \tau {R_{(i,s)}}) -
{U_i}({Q_k})) + {\mathbb{E}}[{{\bf{1}}_{(i,p)}}]({U_i}({Q_k} - \tau
{R_{(i,p)}}) - {U_i}({Q_k}))\}$ is the per-stage objective,
${{\bf{1}}_{(i,p)}} = {\bf{1}}({R_{(i,p)}} \le {C_{(i,p)}})$ and
${{\bf{1}}_{(i,p)}} = {\bf{1}}({R_{(i,s)}} \le {C_{(i,s)}})$ are the
indicator functions.
\end{corollary}

The per-stage optimization above is intractable due to that the
expectation $\mathbb{E}$ required the explicit knowledge of CSIT
errors in BBU pool. To deal with this challenge, the per-stage
optimization problem can be solved by the following stochastic
gradient algorithm{\cite{stochastic}}.
\begin{alg}{(\emph{Stochastic Gradient Algorithm})}

At each frame $t > 1$, the queue-aware power and rate allocations
for each UE can be obtained as the following iteration
\begin{equation}
e_i^{t}({{\bf{\hat S}}_i}) = {[e_i^{t-1}({{\bf{\hat S}}_i}) -
{\gamma _e}(t-1)d(e_i^{t-1}({{\bf{\hat S}}_i}))]^ + },
\end{equation}
where ${\gamma _e}(t)$ is the step size satisfying ${\gamma _e}(t)
> 0,\sum\nolimits_t {{\gamma _e}(t)} = \infty ,\sum\nolimits_t
{{{({\gamma _e}(t))}^2}} < \infty$ and $d(e_i^t({{\bf{\hat S}}_i}))$
is the stochastic gradient w.r.t power and rate allocation, which is
summarized as follows
\begin{equation}
\left\{ \begin{array}{l} \frac{{\partial B({\bf{\hat S}},{\Omega
_P({\bf{\hat S}}),\Omega _R({\bf{\hat S}})})}}{{\partial
P_{(i,p)}^a}} = {\gamma _{(i,P)}} + \frac{{\partial {h_i}({\bf{\hat
S}},{\Omega _P({\bf{\hat S}}),\Omega
_R({\bf{\hat S}})})}}{{\partial P_{(i,p)}^a}}\\
\frac{{\partial B({\bf{\hat S}},{\Omega _P({\bf{\hat S}}),\Omega
_R({\bf{\hat S}})})}}{{\partial P_{(i,s)}^a}} = {\gamma
_{(i,P)}}\rho _{(i,i)}^a + \frac{{\partial {h_i}({\bf{\hat
S}},{\Omega _P({\bf{\hat S}}),\Omega
_R({\bf{\hat S}})})}}{{\partial P_{(i,s)}^a}} \\~~~~~~~~~~~~~~~~~~~~~~ + \sum\limits_{j \ne i,j \in \mathcal {M}} {{\gamma _{(j,P)}}\rho _{(i,j)}^a} \\
\frac{{\partial B({\bf{\hat S}},{\Omega _P({\bf{\hat S}}),\Omega
_R({\bf{\hat S}})})}}{{\partial {R_{(i,p)}}}} = {\gamma _{(i,R)}} +
\frac{{\partial {h_i}({\bf{\hat S}},{\Omega _P({\bf{\hat S}}),\Omega
_R({\bf{\hat S}})})}}{{\partial {R_{(i,p)}}}}\\
\frac{{\partial B({\bf{\hat S}},{\Omega _P({\bf{\hat S}}),\Omega
_R({\bf{\hat S}})})}}{{\partial {R_{(i,p)}}}} = {\gamma _{(i,R)}} +
\frac{{\partial {h_i}({\bf{\hat S}},{\Omega _P({\bf{\hat S}}),\Omega
_R({\bf{\hat S}})})}}{{\partial {R_{(i,p)}}}} \\ ~~~~~~~~~~~~~~~~~
~~~~~+ \sum\limits_{j \ne i,j \in \mathcal {M}} {{\gamma _{(j,R)}}}
\end{array} \right.\label{stoch},
\end{equation}
where ${h_i}({\bf{\hat S}},{\Omega _P({\bf{\hat S}}),\Omega
_R({\bf{\hat S}})}) = {{\bf{1}}_{(i,s)}}({U_i}({Q_k} - \tau
{R_{(i,s)}}) - {U_i}({Q_k})) + {{\bf{1}}_{(i,p)}}({U_i}({Q_k} - \tau
{R_{(i,p)}}) - {U_i}({Q_k})) + {{\bf{1}}_{(i,s)}}{{\bf{1}}_{(i,p)}}(
{U_i}({Q_k} - \tau {R_{(i,s)}} - \tau {R_{(i,p)}}) - {U_i}({Q_k} -
\tau {R_{(i,s)}}) - {U_i}({Q_k} - \tau {R_{(i,p)}}) +
{U_i}({Q_k}))$.
\end{alg}

When ${\bf{\hat H}} = {\bf{ H}}$, there is no interference under the
H-CoMP with perfect CSIT, and $d(e_i^t({{\bf{\hat S}}_i}))$ is
deterministic instead of stochastic. Using the standard gradient
update argument, the gradient search converges to a local optimum as
$t \to \infty$. Therefore, the Algorithm 1 gives the asymptotically
local optimal solution at small CSIT errors, which means that the
explicit knowledge of imperfect CSIT is unnecessary and it is robust
against the uncertainties caused by imperfect CSIT.

\begin{remark}{(\emph{Feedback-Assisted Realization of Algorithm 1})}
The calculation of stochastic gradient (\ref{stoch}) in BBU pool
requires some items regarding the indicator functions
${\bf{1}}_{(i,s)}$ and ${\bf{1}}_{(i,p)}$ and the differential of
post-decision value functions $ U_i^{'}({\tilde Q_i})$. At each
frame $t$, the indicator functions are unknown by BBU pool and have
to be fed back from UEs, which is feasible due to that there are
existing built-in mechanisms in wireless networks for these ACK/NACK
feedback from UEs. In addition, since there is no closed-form
expression of post-decision value function $U_i^{'}({\tilde Q_i})$,
its differential can be estimated as follows
\begin{equation}
U_i^{'}({\tilde Q_i}) = U_i({\tilde Q_i}) - U_i({\tilde Q_i} - 1),
\end{equation}
where the online learning of $U_i({\tilde Q_i})$ will be elaborated
in next subsection.
\end{remark}

\subsection{Online Learning of Per-queue Post-decision Value Functions}

The post-decision value functions are critical to the derivation of
queue-aware resource allocation policy for C-RANs, which can be
obtained by solving $N_Q$ fixed point nonlinear Bellman equations
with $N_Q + 1$ variables. The offline calculation requires the
explicit knowledge of conditional average transition kernel, which
is infeasible. In this section, with the realtime observation of QSI
and CSIT, the online learning of per-queue post-decision value
functions is proposed based on the equation (\ref{perqueuebellman}).
Meanwhile, with the realtime resource control actions, the LMs are
updated to make sure the average power consumption constraints and
average fronthaul consumption constraints are
satisfied{\cite{xrcao}}. The online learning of per-queue value
functions and the update of LMs at centralized BBU pool are
described as follows.

\begin{alg}{(\emph{Online Learning of Per-Queue Value Functions and Update of
LMs})}
\begin{step}[\textbf{Initialization}]
Set $t = 0$, the per-queue post-decision value functions
$\{U_i^{0}({{\tilde Q}_i})\}$ and LMs $\{\gamma _{(i,P)}^{0}, \gamma
_{(i,R)}^{0} \} > 0$ are initialized at the centralized BBU pool.
\end{step}
\begin{step}[\textbf{Queue-Aware Resource Allocation}]
At the beginning of the $t$-th frame, given fixed $\{ {\eta
_j}(n)\}$, the queue-aware power and rate allocation for downlink
H-CoMP transmission are determined at the BBU pool using the
stochastic gradient algorithm in (\ref{stoch}).
\end{step}
\begin{step}[\textbf{Online Learning of $U_i^{t + 1}({{\tilde Q}_i})$}]
With the observation of post-decision QSI $\{{{\tilde Q}_i} \}$ and
pre-decision QSI $\{{Q_i}\}$, the per-queue post-decision value
function $U_i({{\tilde Q}_i})$ is online learned at BBU pool
(\ref{valueupdate}) for each traffic queue as follows
\begin{eqnarray}
U_i^{t + 1}({{\tilde Q}_i})&& = U_i^t({{\tilde Q}_i}) + {\zeta
_u}(t)[{g_i}(\gamma _i^t,{{{\bf{\hat S}}}_i},{P_i},{R_i} +
U_i^t({Q_i} - {U_i})\nonumber\\&& - U_i^t(\tilde Q_i^0) -
U_i^t({{\tilde Q}_i})]\label{valueupdate}.
\end{eqnarray}
\end{step}
\begin{step}[\textbf{Update of $\{\gamma
_{(i,P)}^{t + 1}, \gamma _{(i,R)}^{t + 1} \}$}] With the observation
of power and rate allocation, the $\{\gamma _{(i,P)}^{t + 1}\}$ and
$\{ \gamma _{(i,R)}^{t + 1} \}$ are updated according to
eq.(\ref{lmpower}) and eq.(\ref{lmrate}) at the BBU pool for the
power consumption constraints and fronthaul consumption constraints
respectively,
\begin{equation}
\gamma _{(i,P)}^{t + 1} \!= \!{[\gamma _{(i,P)}^t + {\zeta _\gamma
}\!(t)({P_i} - P_i^{\max })]^ + }\label{lmpower},
\end{equation}
\begin{equation}
\gamma _{(i,R)}^{t + 1}\! = \!{[\gamma _{(i,R)}^t\! + \!{\zeta
_\gamma }\!(t)({R_{(i,p)}}(t)\! + \!\sum\nolimits_{j \in \mathcal
{M}} {{R_{(j,s)}}\!(t)} \! -\! R_i^{\max })]^ + }\label{lmrate}.
\end{equation}
\end{step}
\begin{step}[\textbf{Termination}]
Set $t = t + 1$ and continue to step 2 until certain termination
condition is satisfied.
\end{step}

The ${\zeta _u}(t)$ and ${\zeta _\gamma }(t)$ in step 3 and step 4
is the iterative step size of post-decision value functions and LMs
respectively. To make sure the convergence of iteration, they should
satisfy the conditions as follows{\cite{convergence}}:

\begin{equation}
 {\zeta _u}(t) > 0,\sum\nolimits_t {{\zeta _u}(t)} = \infty,
\end{equation}
\begin{equation}
 {\zeta _\gamma }(t) > 0,\sum\nolimits_t {{\zeta _\gamma }(t)} =
 \infty,
\end{equation}
\begin{equation}
 \sum\nolimits_t{({{({\zeta _u}(t))}^2} + } {({\zeta _\gamma }(t))^2}) <
 \infty,
\end{equation}
\begin{equation}
\mathop {\lim }\limits_{t \to \infty } \frac{{{\zeta _\gamma
}(t)}}{{{\zeta _u}(t)}} = 0\label{stepsize}.
\end{equation}
\end{alg}

\begin{remark}[Two Timescales of Iterations]
The condition (\ref{stepsize}) implies that the LMs are relatively
static during the iteration of per-queue value functions. Therefore
the iteration of post-decision value functions and the iteration of
LMs are done simultaneously but over two different time
scales{\cite{onlinelearning}}.
\end{remark}

It is a remarkable fact that the size of per-queue states(in bits)
is still large. To accelerate the estimation of each post-state
value function, the per-queue QSI space ${{\mathcal {Q}}_i}$ is
partitioned into $N$ regions as (\ref{region})

\begin{equation}
{{\mathcal {Q}}_i} = \bigcup\nolimits_{n = 1}^N {{{\mathcal
{R}}_n}}\label{region}.
\end{equation}
Therefore, the average value function w.r.t each region is online
learned instead, then the post-decision value function of each state
within the region can be estimated by interpolation method after
each iteration.

\section{Performances Evaluation}

In this section, simulations are conducted to compare the
performances of the proposed QAH-CoMP with various baselines in
C-RANs. The delay-sensitive traffic packet arrival follows a Poisson
distribution and the corresponding packet size follows an
exponential distribution, which is a widely adopted traffic
model{\cite{survey}}. The mean size of traffic packet is 4Mbits and
the maximum buffer size is 32Mbits. The CSI ${{\bf{H}}_{ij}}$ is
uniformly distributed over a state space ${\mathcal {H}}^{N_r \times
N_t}$ and the error variance of the imperfect CSIT is ${\varepsilon
_e} = 0.05$. The configuration of multi-antennas is given by $\{
{N_t} = 5,{N_r} = 2\}$ and the cluster size is $M$ = 3. Therefore,
with the stream splitting of the H-CoMP scheme, there are one shared
stream and one private stream to be transmitted for each UE. The
total bandwidth of simulated C-RAN is 20MHz and the scheduling frame
duration is 10ms. The noise power is -15dBm.

Three baselines are considered in the simulations: CB-CoMP, JP-CoMP,
and channel-aware resource allocation with H-CoMP (CAH-CoMP). All
these three baselines carry out rate and power allocation to
maximize the average system throughput with the same fronthaul
capacity and average power consumption constraint as the proposed
QAH-CoMP. For the CB-CoMP baseline, the BBU pool calculates the
coordinated beamformer for each RRH to eliminate the dominating
intra-cluster interference. For the CAH-CoMP baseline, the proposed
H-CoMP transmission is adopted, while the power allocation and rate
allocation are only adaptive to CSIT.

Fig. 3 compares the delay performance of the four schemes with
different packet arrival rate. The average packet delay of all the
schemes increases as the average packet arrival rate increases.
Compared with CB, the delay outperformance of JP weakens as the
packet arrival rate increases, which is due to the fact that the
fronthaul capacity becomes relatively limited with the increasing
packet arrival rate. Apparently, the performance gain of QAH-CoMP
compared with CAH-CoMP is contributed by power and rate allocation
with the consideration of both urgent traffic flows and imperfect
CSIT.

\begin{figure}
\centering \vspace*{0pt}
\includegraphics[scale=0.5]{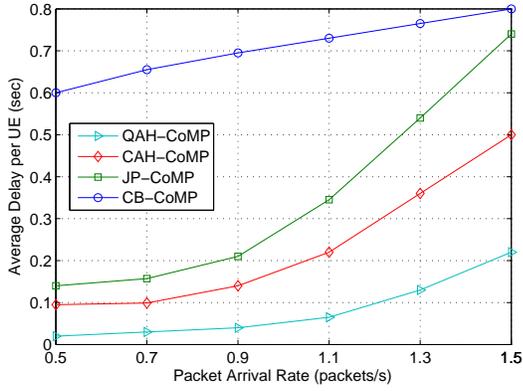}
\setlength{\belowcaptionskip}{-100pt}
\vspace*{-10pt}\caption{Average packet delay vs. packet arrival
rate, the maximum fronthaul consumption is $R_i^{max}$ = 20Mbits/s,
the maximum transmit power is $P_i^{max}$ = 10dBm.} \label{fig2}
\end{figure}

Fig. 4 compares the delay performance of the four schemes with
different maximum transmit power. The figure depicts the medium
fronthaul consumption regime, in which JP-CoMP outperforms CS-CoMP
due to the higher spectrum efficiency. CAH-CoMP outperforms both
CS-CoMP and JP-CoMP while the outperformance of CAH-CoMP is not so
obvious with relative enough fronthaul capacity. It can be observed
that there is significant performance gain of the proposed QAH-CoMP
compared with all the baselines across a wide range of the maximum
power consumption.

\begin{figure}
\centering \vspace*{0pt}
\includegraphics[scale=0.5]{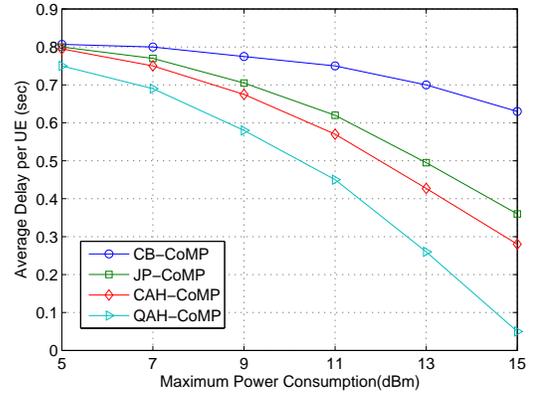}
\setlength{\belowcaptionskip}{-100pt}
\vspace*{-10pt}\caption{Average packet delay vs. maximum transmit
power, the packet arrival rate is ${\lambda}_i$ = 2.5 packets/s, the
maximum fronthaul consumption is $R_i^{max}$ = 30Mbits/s.}
\label{fig3}
\end{figure}

Fig. 5 compares the delay performance of the four schemes with
different maximum fronthaul consumption. The figure depicts the
small fronthaul consumption regime, in which CAH-CoMP clearly
outperforms both CS-CoMP and JP-CoMP, which is contributed by the
flexible adjustment of cooperation level when the fronthaul capacity
is limited. Note that the JP-CoMP has worse delay performance than
CS-CoMP due to limited fronthaul capacity at first but it eventually
gets performance improvement with increasing fronthaul capacity.
Similarly, due to the queue-aware power and rate allocation,
QAH-CoMP substantially outperforms the three baselines.

\begin{figure}
\centering \vspace*{0pt}
\includegraphics[scale=0.5]{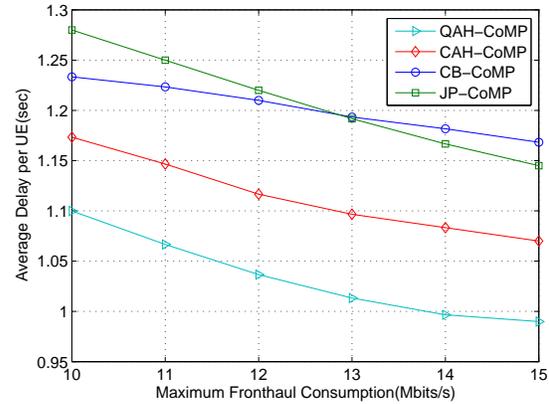}
\setlength{\belowcaptionskip}{-100pt} \vspace*{-2pt}\caption{Average
packet delay vs. maximum fronthaul consumption,the packet arrival
rate is ${\lambda}_i$ = 2.5 packets/s, the maximum transmit power is
$P_i^{max}$ = 10dBm.} \label{fig4}
\end{figure}

Fig. 6 shows the convergence property of the online per-queue
post-decision value functions(w.r.t ${{\mathcal {R}}_n}$, size of
which is equal to mean packet size ${\bar N_i}$) learning algorithm.
For viewing convenience, the post-decision value functions of the
traffic queue maintained for UE 1 is plotted with the increasing of
iteration step. It is significant that the learning converges
extremely close to the final result after 1000 iterations.

\begin{figure}
\centering \vspace*{0pt}
\includegraphics[scale=0.5]{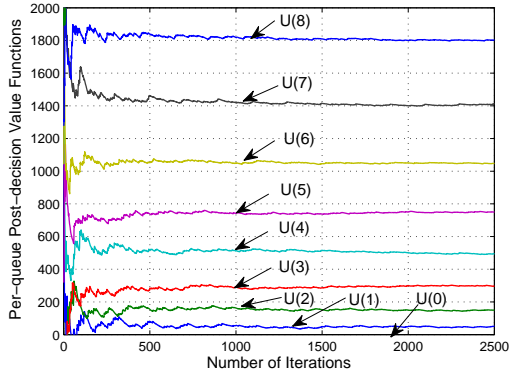}
\setlength{\belowcaptionskip}{-100pt}
\vspace*{-10pt}\caption{Per-queue Post-decision Value Functions}
\label{fig5}
\end{figure}

\section{Summary}

In this paper, an H-CoMP scheme with corresponding precoders and
decorrelators are designed for the downlink fronthaul constrained
C-RANs. Based on the proposed H-CoMP, a low complexity queue-aware
power and rate allocation solution for the delay-sensitive traffic
is then proposed using MDP and stochastic gradient algorithms.
Simulation results show that the C-RANs with H-CoMP achieve more
significant delay performance gains than that with CB-CoMP and
JP-CoMP under the same average power and fronthaul consumption
constraints, where the performance gains largely depend on the
cooperation level of the proposed H-CoMP under limited fronthaul
capacity. Furthermore, compared with the CAH-CoMP, the remarkable
delay performance gain of QAH-CoMP is also validated by the
simulation results, which is contributed by the MDP based dynamic
resource allocation with the consideration of both QSI and imperfect
CSIT. In the future, the theoretical analysis on the delay
performance of the QAH-CoMP still remains to be an open issue, and
the real experiments would be desirable to further demonstrate the
effectiveness and applicability of the QAH-CoMP in fronthaul
constrained C-RANs.

\appendices
\section{Proof of theorem 1}
According to the Proposition 4.6.1 of{\cite{bertsekas}}, the
sufficient condition for optimality of problem 1 is that there
exists unique $(\theta ,\{ U({\bf{S}})\} )$ that satisfies the
following Bellman equation and $U({\bf{S}})$ satisfies the
transversality condition $\mathop {\lim }\limits_{T \to \infty }
\frac{1}{T}{{\mathbb{E}}^\Omega }[U({\bf{S}}(T))|{\bf{S}}(0)] = 0$
for all admissible control policy $\Omega$ and initial state
${\bf{S}}(0)$.
\begin{equation}
\begin{array}{l}
 \theta + U({\bf{S}}) = \mathop {\min }\limits_{\Omega ({\bf{\hat S}})} [g({\bf{\beta }},{\bf{\gamma }},{\bf{S}},\Omega ({\bf{\hat S}})) + \sum\limits_{{\bf{S}}'} {\Pr [{\bf{S}}|{\bf{S}},\Omega ({\bf{\hat S}})]} U({\bf{S}}')] \\
 = \mathop {\min }\limits_{\Omega ({\bf{\hat S}})} [g({\bf{\beta }},{\bf{\gamma }},{\bf{S}},\Omega ({\bf{\hat S}})) + \sum\limits_{{\bf{Q}}'} {\sum\limits_{{\bf{\hat H}}'} {\sum\limits_{{\bf{H}}'} {\Pr [{\bf{Q}}'|{\bf{Q}},{\bf{\hat H}},{\bf{H}},\Omega ({\bf{\hat S}})]}}} \\{{\Pr [{\bf{\hat H}}',{\bf{H}}']U({\bf{S}}')]} } \\
 \end{array}
\end{equation}

Then taking expectation w.r.t. ${\bf{\hat H}}',{\bf{H}}'$ on both
side of the above equation, we have

\begin{equation}
\theta \! + \!U({\bf{Q}})\! =\! \mathop {\min }\limits_{\Omega
({\bf{Q}})} [g({\bf{\beta }},{\bf{\gamma }},{\bf{Q}},\Omega
({\bf{Q}})) + \sum\limits_{{\bf{Q}}'} {\Pr
[{\bf{Q}}'|{\bf{Q}},\Omega ({\bf{Q}})]U({\bf{Q}}')} ]
\end{equation}
where $U({\bf{Q}}) = {\mathbb{E}}[U({\bf{S}})|{\bf{Q}}] =
\sum\limits_{{\bf{\hat H}}} {\sum\limits_{\bf{H}} {\Pr [{\bf{\hat
H}},{\bf{H}}]U({\bf{S}})} }$ and $\Pr [{\bf{Q}}'|{\bf{Q}},\Omega
({\bf{\hat S}})] = {\mathbb{E}}[\Pr [{\bf{Q}}'|{\bf{Q}},{\bf{\hat
H}},{\bf{H}},\Omega ({\bf{\hat S}})]|{\bf{Q}}]$. Since here we
defined the post-decision State ${\bf{\tilde Q}}$, where ${\bf{Q}} =
\min \{ {\bf{\tilde Q}} + {\bf{A}},{N_Q}\}$, the equivalent Bellman
equation can be transformed as the equivalent Bellman equation
(\ref{bellman}) in theorem 1.

\section{Proof of lemma 1}

With the perfect CSIT, there is no interference with the H-CoMP
scheme for C-RAN, which means that the queue dynamics for every UE
are completely decoupled. Detailedly speaking, ${\tilde Q_i} = {Q_i}
- {G_i}({\bf{\hat H}},{\Omega _i}({\bf{\hat S}}))$ is independent of
${Q_j}$ and ${\Omega _j}({\bf{\hat S}})$ for all $j \ne i$ due to
the nonexistence of interference, therefore we have $\Pr
[{\bf{\tilde Q}}'|{\bf{Q}},\Omega ({\bf{Q}})] = \prod\nolimits_{i
\in \mathcal {M}} {\Pr [{{\tilde Q}_i}'|{\bf{Q}},\Omega ({\bf{Q}})]}
$ and $\Pr [{\tilde Q_i}'|{\bf{Q}},\Omega ({\bf{Q}})] = \Pr [{\tilde
Q_i}'|{Q_i},{\Omega _i}({\bf{Q}})] = \Pr [{\tilde
Q_i}'|{Q_i},{\Omega _i}({Q_i})]$. Suppose $U({\bf{\tilde Q}}) =
\sum\nolimits_{i \in \mathcal {M}} {{U_i}({{\tilde Q}_i})}$, by the
relationship between the joint distribution and the marginal
distribution, we have
\begin{equation}
\begin{array}{l}
 {~~~}\sum\nolimits_{{\bf{\tilde Q}}'} {\Pr [{\bf{\tilde Q}}'|{\bf{Q}},\Omega ({\bf{Q}})]U({\bf{\tilde Q}}')} \\
 = \sum\nolimits_{{\bf{\tilde Q}}'} {\Pr [{\bf{\tilde Q}}'|{\bf{Q}},\Omega ({\bf{Q}})]\sum\nolimits_{i \in M} {{U_i}({{\tilde Q}_i}')} } \\
 = \sum\nolimits_{i \in \mathcal {M}} {\sum\nolimits_{{{\tilde Q}_i}'} {\Pr [{{\tilde Q}_i}'|{\bf{Q}},\Omega ({\bf{Q}})]} } {U_i}({{\tilde Q}_i}') \\
 = \sum\nolimits_{i \in \mathcal {M}} {\sum\nolimits_{{{\tilde Q}_i}'}}
 {{\Pr [{{\tilde Q}_i}'|{Q_i},{\Omega _i}({Q_i})]{U_i}({{\tilde Q}_i}')} } \\
 \end{array}
\end{equation}

It is obvious that $g({\bf{\beta }},{\bf{\gamma }},{\bf{Q}},\Omega
(\Theta )) = \sum\nolimits_{i \in M} {{g_i}({\beta _i},{{\bf{\gamma
}}_i},{Q_i},{\Omega _i}({Q_i}))}$. Suppose $\theta =
\sum\nolimits_{i \in \mathcal {M}} {{\theta _i}}$, then the
equivalent Bellman equation in (\ref{bellman}) can be transformed as

\begin{equation}
\begin{array}{l}
 {~~~~}\sum\nolimits_{i \in \mathcal {M}} {{\theta _i}} + \sum\nolimits_{i \in \mathcal {M}} {{U_i}({{\tilde Q}_i})} \\
 = \sum\nolimits_{\bf{A}} {\Pr ({\bf{A}})} \mathop {\min }\limits_{\Omega (Q)} \sum\nolimits_{i \in M} {[{g_i}({\beta _i},{{\bf{\gamma }}_i},{Q_i},{\Omega _i}({Q_i}))}\\~~~{ + \sum\nolimits_{{{\tilde Q}_i}'} {\Pr [{{\tilde Q}_i}'|{Q_i},{\Omega _i}({Q_i})]{U_i}({{\tilde Q}_i}')]} } \\
 \mathop = \limits^{(a)} \sum\nolimits_{i \in \mathcal {M}} {\sum\nolimits_{{A_i}} {\Pr ({A_i})} } \mathop {\min }\limits_{{\Omega _i}({Q_i})} [{g_i}({\beta _i},{{\bf{\gamma }}_i},{Q_i},{\Omega _i}({Q_i})) \\~~~+ \sum\nolimits_{{{\tilde Q}_i}'} {\Pr [{{\tilde Q}_i}'|{Q_i},{\Omega _i}({Q_i})]{U_i}({{\tilde Q}_i}')]} \\
 \end{array}
\end{equation}
where (a) is due to the independent assumption of the new arrival
process $A_i(t)$ w.r.t $i$. Therefore, we can have the per-queue
fixed point Bellman equation in (\ref{perqueuebellman}) for each UE
from the above equation.

\begin{IEEEbiography}[{\includegraphics[width=1in,height=1.2in]{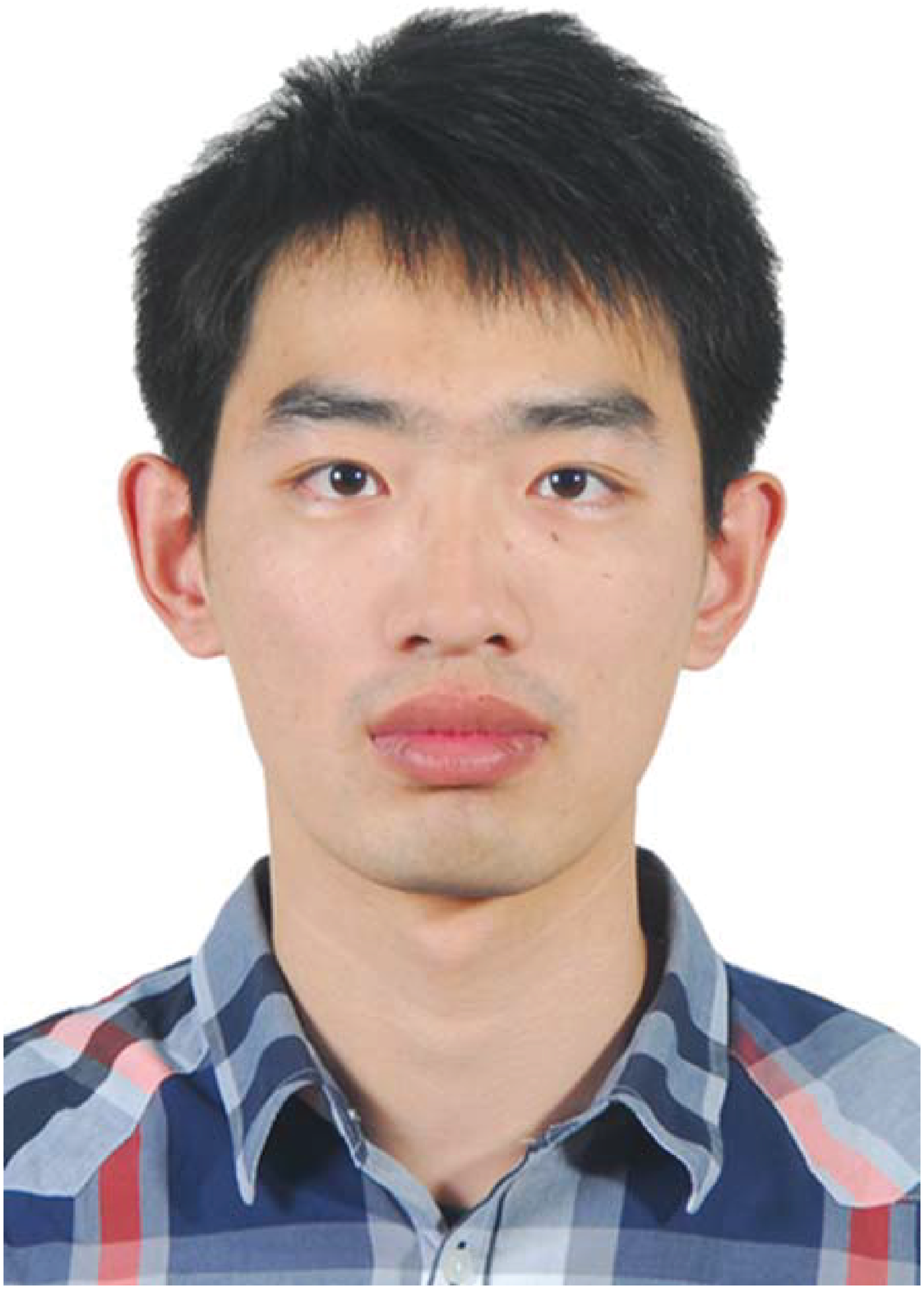}}]{Jian Li}
received his B.E. degree from Nanjing University of Posts and
Telecommunications, Nanjing, China, in 2010. He is currently
pursuing his Ph.D. degree at the key laboratory of universal
wireless communication (Ministry of Education) in Beijing University
of Posts and Telecommunications (BUPT), Beijing, China. His current
research interests include delay-aware cross-layer radio resource
optimization for heterogeneous networks and heterogeneous cloud
radio access networks.
\end{IEEEbiography}
\begin{IEEEbiography}[{\includegraphics[width=1in,height=1.25in]{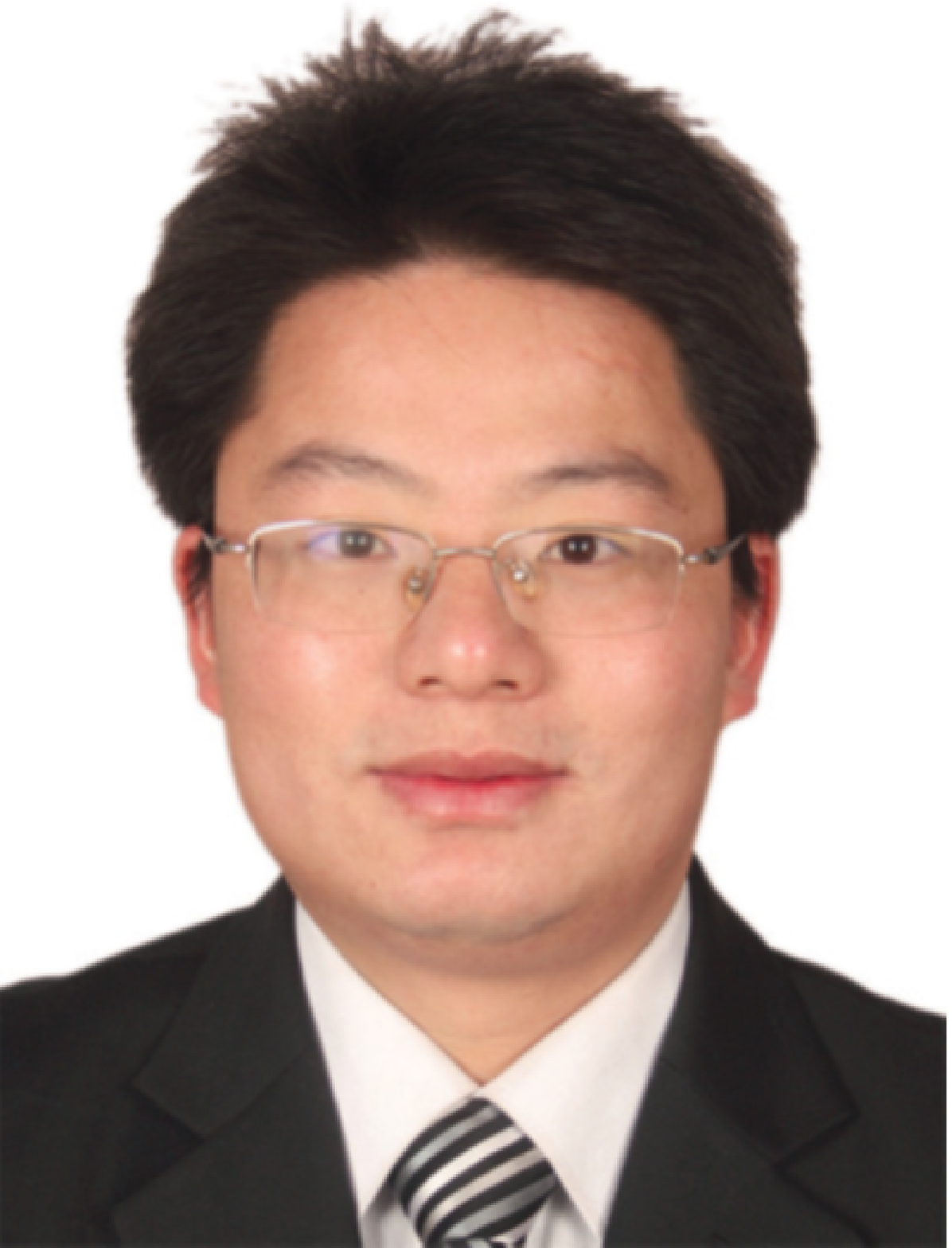}}]{Mugen Peng}
(M'05--SM'11) received the B.E. degree in Electronics Engineering
from Nanjing University of Posts \& Telecommunications, China in
2000 and a PhD degree in Communication and Information System from
the Beijing University of Posts \& Telecommunications (BUPT), China
in 2005. After the PhD graduation, he joined in BUPT, and has become
a full professor with the school of information and communication
engineering in BUPT since Oct. 2012. During 2014, he is also an
academic visiting fellow in Princeton University, USA. He is leading
a research group focusing on wireless transmission and networking
technologies in the Key Laboratory of Universal Wireless
Communications (Ministry of Education) at BUPT, China. His main
research areas include wireless communication theory, radio signal
processing and convex optimizations, with particular interests in
cooperative communication, radio network coding, self-organization
networking, heterogeneous networking, and cloud communication. He
has authored/coauthored over 40 refereed IEEE journal papers and
over 200 conference proceeding papers.

Dr. Peng is currently on the Editorial/Associate Editorial Board of
IEEE Communications Magazine, IEEE Access, International Journal of
Antennas and Propagation (IJAP), China Communication, and
International Journal of Communication Systems (IJCS). He has been
the guest leading editor for the special issues in IEEE Wireless
Communications, IJAP and the International Journal of Distributed
Sensor Networks (IJDSN). He is serving as the track co-chair or
workshop co-chair for GameNets 2014, So-HetNets in IEEE WCNC 2014,
SON-HetNet 2013 in IEEE PIMRC 2013, WCSP 2013, etc. Dr. Peng was
honored with the Best Paper Award in CIT 2014, ICCTA 2011, IC-BNMT
2010, and IET CCWMC 2009. He was awarded the First Grade Award of
Technological Invention Award in Ministry of Education of China for
his excellent research work on the hierarchical cooperative
communication theory and technologies, and the Second Grade Award of
Scientific \& Technical Progress from China Institute of
Communications for his excellent research work on the co-existence
of multi-radio access networks and the 3G spectrum management in
China.
\end{IEEEbiography}
\begin{IEEEbiography}[{\includegraphics[width=1in,height=1.25in]{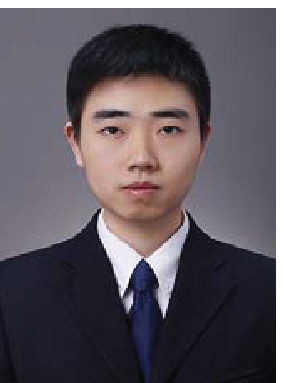}}]{Aolin Cheng}
received his B.E. degree in Electronic Information Science and
Technology from Beijing University of Posts and Telecommunications,
Beijing, China, in 2012. He is currently pursuing his M.E. degree at
the laboratory of universal wireless communication (Ministry of
Education) in Beijing University of Posts and Telecommunications
(BUPT), Beijing, China. His current research interests include
delay-aware cross-layer radio resource optimization for
heterogeneous networks (HetNets), as well as stochastic
approximation and Markov decision process.
\end{IEEEbiography}
\begin{IEEEbiography}[{\includegraphics[width=1in,height=1.25in]{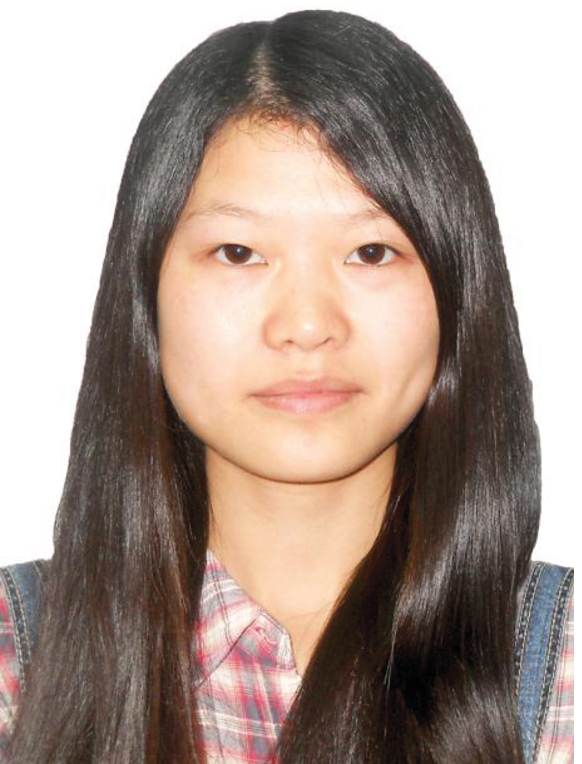}}]{Yuling Yu}
received the B.E. degree in Communication Engineering from Wuhan
University of Technology, Wuhan, China, in 2013. She is currently
pursuing her M.E. degree at the key laboratory of universal wireless
communication (Ministry of Education) in Beijing University of Posts
and Telecommunications (BUPT), Beijing, China. Her research focuses
on delay-aware cross-layer resource optimization for heterogeneous
cloud radio access networks (H-CRANs), as well as Lyapunov
optimization.
\end{IEEEbiography}
\begin{IEEEbiography}[{\includegraphics[width=1in,height=1.25in]{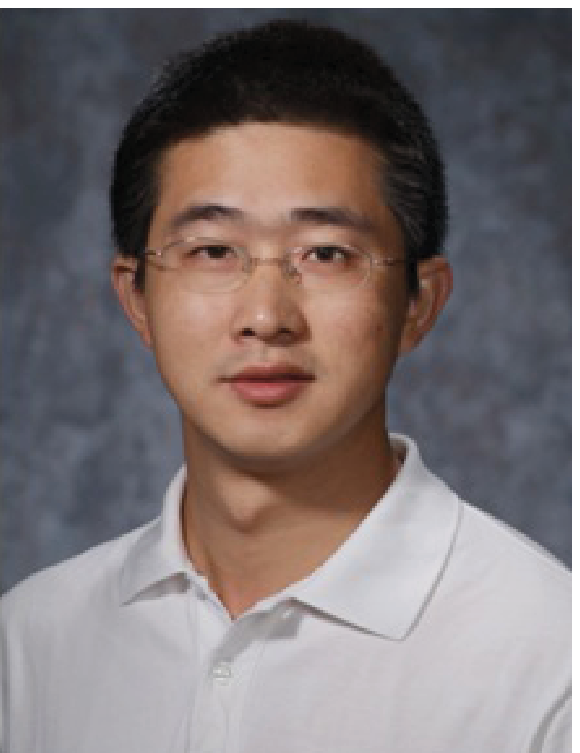}}]{Chonggang Wang}
(SM'09) received the Ph.D. degree from Beijing University of Posts
and Telecommunications (BUPT) in 2002. He is currently a Member of
Technical Staff with InterDigital Communications. His R\&D focuses
on: Internet of Things (IoT), Machine-to-Machine (M2M)
communications, Heterogeneous Networks, and Future Internet,
including technology development and standardization. He
(co-)authored more than 100 journal/conference articles and book
chapters. He is on the editorial board for several journals
including IEEE Communications Magazine, IEEE Wireless Communications
Magazine and IEEE Transactions on Network and Service Management. He
is the founding Editor-in-Chief of IEEE Internet of Things Journal.
He is serving and served in the organization committee for
conferences/workshops including IEEE WCNC 2013, IEEE INFOCOM 2012,
IEEE Globecom 2010-2012, IEEE CCNC 2012, and IEEE SmartGridComm
2012. He has also served as a TPC member for numerous conferences
such as IEEE ICNP (2010-2011), IEEE INFOCOM (2008-2014), IEEE
GLOBECOM (2006-2014), IEEE ICC (2007-2013), IEEE WCNC (2008-2012)
and IEEE PIMRC (2012-2013). He is a co-recipient of National Award
for Science and Technology Achievement in Telecommunications in 2004
on IP QoS from China Institute of Communications. He received
Outstanding Leadership Award from IEEE GLOBECOM 2010 and
InterDigital's 2012 and 2013 Innovation Award. He served as an NSF
panelist in wireless networks in 2012. He is a senior member of the
IEEE and the vice-chair of IEEE ComSoc Multimedia Technical
Committee (MMTC) (2012-2014).
\end{IEEEbiography}

\end{document}